\documentclass[conference]{IEEEtran}
\usepackage{blindtext, graphicx}
\usepackage{cuted}

\usepackage{bigints}
\usepackage{caption}
\newlength\figureheight 
\newlength\figurewidth 
\usepackage{cite}
\usepackage{url}
\usepackage{blindtext, graphicx}
\usepackage{amsmath,amssymb,mathtools,bm,etoolbox}
\newcommand{\e}[1]{{\mathbb E}\left[ #1 \right]}
\usepackage{stackengine}
\def\delequal{\mathrel{\ensurestackMath{\stackon[1pt]{=}{\scriptstyle\Delta}}}}
\usepackage{suffix}
\DeclarePairedDelimiterX\MeijerM[3]{\lparen}{\rparen}%
{\begin{smallmatrix}#1 \\ #2\end{smallmatrix}\delimsize\vert\,#3}

\newcommand\MeijerG[8][]{%
  G^{\,#2,#3}_{#4,#5}\MeijerM[#1]{#6}{#7}{#8}}

\WithSuffix\newcommand\MeijerG*[7]{%
  G^{\,#1,#2}_{#3,#4}\MeijerM*{#5}{#6}{#7}}
\newcommand\Mark[1]{\textsuperscript#1}

\begin{document}
%
% paper title
% can use linebreaks \\ within to get better formatting as desired
\title{Partial Relay Selection For Hybrid RF/FSO Systems with Hardware Impairments}

% author names and affiliations
% use a multiple column layout for up to three different
% affiliations
\author{Elyes Balti\Mark{1}, Mohsen Guizani\Mark{1}, Bechir Hamdaoui\Mark{2} and Yassine Maalej\Mark{1}\\
\Mark{1}University of Idaho, USA,  \Mark{2}Oregon State University, USA}

% make the title area
\maketitle

\begin{abstract}
%\boldmath
In this paper, we investigate the performance analysis of dual hop relaying system consisting of asymmetric Radio Frequency (RF)/Free Optical Space (FSO) links. The RF channels follow a Rayleigh distribution and the optical links are subject to Gamma-Gamma fading. We also introduce impairments to our model and we suggest Partial Relay Selection (PRS) protocol with Amplify-and-Forward (AF) fixed gain relaying. The benefits of employing optical communication with RF, is to increase the system transfer rate and thus improving the system bandwidth. Many previous research attempts assuming ideal hardware (source, relays, etc.) without impairments. In fact, this assumption is still valid for low-rate systems. However, these hardware impairments can no longer be neglected for high-rate systems in order to get consistent results. Novel analytical expressions of outage probability and ergodic capacity of our model are derived taking into account ideal and non-ideal hardware cases. Furthermore, we study the dependence of the outage probability and the system capacity considering, the effect of the correlation between the outdated CSI (Channel State Information) and the current source-relay link, the number of relays, the rank of the selected relay and the average optical Signal to Noise Ratio (SNR) over weak and strong atmospheric turbulence. We also demonstrate that for a non-ideal case, the end-to-end Signal  to Noise plus Distortion Ratio (SNDR) has a certain ceiling for high SNR range. However, the SNDR grows infinitely for the ideal case and the ceiling caused by impairments no longer exists. Finally, numerical and simulation results are presented.
\end{abstract}


\IEEEpeerreviewmaketitle

\section{Introduction}
Optical wireless communication technology has emerged as a promising solution to assist the Radio Frequency part of a relaying system that reaches the bottleneck in terms of bandwidth efficiency, power consumption and transfer rate. Such a system is called a hybrid RF/FSO. The vast majority of research work in this area investigated the performance of such systems as the outage probability, bit error rate and ergodic capacity for many scenarios, such as systems with a single relay \cite{2} or multiple relays \cite{16}, \cite{20} employing different relaying schemes like Amplify-and-Forward (AF) \cite{13}, \cite{14}, \cite{15}, Decode-and-Forward (DF) \cite{16}, \cite{17} and Quantize-and-Forward (QF) \cite{18}, \cite{19}. Also, mixed RF/FSO systems with multiple relays employ many relay selection protocols like best and opportunistic relay selection \cite{21}, \cite{22}, \cite{28}, partial relay selection \cite{21}, \cite{22}, \cite{6} (selecting the relay based on the knowledge of RF or FSO CSI), Distributed Switch and Stay protocol \cite{28}, \cite{23} (selecting the relay that belongs to the shortest RF/FSO link) and all-active relaying \cite{23} which consists of transmitting simultaneously on the overall links. Systems employing the latest selection protocol suffer from the problem of synchronization at the reception since it employs optical signals, that is why such a system configuration is no longer used in practice.\\
In reality, the channels are time-varying and so the instantaneous CSI used for selecting the relay is outdated because of slow feedback propagation from the relays to the source or the destination. In this case, the selection of the best relay is of low probability. Hence, the system performance will deteriorate due to the outdated CSI.\\
Although these works suggested systems without taking into account hardware impairments, this assumption is still valid for low-rate systems. However, hybrid RF/FSO systems are characterized by high transfer rates due to the optical contributions so the assumption of neglecting hardware impairments is no longer valid since for high SNR these impairments significantly affect the system performance.\\
At the practical level, hardware suffers from many impairment types like I/Q imbalance \cite{11} which mitigates the magnitude and shifts the signal phase. Other impairments like High Power Amplifier (HPA) non linearities \cite{12}, \cite{hpa} cause the signal distortion. Qi \textit{et al.} \cite{24} concluded that the hardware impairments have a destructive effect on the system performance. In addition, \cite{1} introduced a general model of hardware impairment accounting for many types of impairment unlike \cite{11}, \cite{12} that considered only one type of impairment. Although Bjornson \textit{et al.} modeled the overall impairments into one aggregate type in \cite{1}, \cite{7}, they applied this model on fully RF systems with a single relay.
Our contribution is to consider this aggregate impairment model with a hybrid RF/FSO system of multiple relays employing Amplify-and-Forward (AF) with fixed gain relaying. Our system also employs PRS protocol based on the knowledge of the first hop to select the qualified relay to forward the signal to the destination. To the best of our knowledge, we are the first research group to employ hardware impairments with a mixed RF/FSO system of multiple relays employing partial relay selection based on the knowledge of the CSI of the first hop.\\
The rest of this paper is organized as follows. Section II describes the system model. The outage probabiliy and ergodic capacity analysis are presented in section III. Section IV presents the numerical results and their discussions. The final section presents some concluding remarks and future research directions.
\section{System Model}
Our system is composed of source (S), destination (D) and $N$ parallel relays linked to (S) and (D). Those relays employ AF with fixed gain and we refer to the PRS with outdated CSI to select the relay of rank $m$ to forward the communication. In addition, we assume Intensity Modulation and Direct Detection (IM/DD) at the reception.\\
For each transmission, the source (S) receives local feedbacks from $N$ relays and arranges the received instantaneous SNR ($\tilde{\gamma}_{1(i)}$ for $i$ = $1$,\ldots $N$) of the first hop in an increasing order of magnitude as follows: $\tilde{\gamma}_{1(1)}\leq\tilde{\gamma}_{1(2)}\leq \ldots \leq\tilde{\gamma}_{1(N)}$.\\
The best scenario is to select the best relay ($m = N$) with the highest first hop SNR. However, the best relay is not always available for forwarding, in this case the source will select the next best available relay. Due to the imperfect estimation of the CSI, the outdated CSI used for PRS $\tilde{\gamma}_{1(m)}$ and the actual RF CSI $\gamma_{1(m)}$ used for forwarding the signal are correlated with correlation coefficient $\rho$.  
%\begin{center}
%\includegraphics[width=8.5cm,height=6cm]{prs}
%\captionof{figure}{Hybrid RF/FSO System}
%\label{fig1}
%\end{center}
The received signal at the $m$-th relay is given by:\\
\begin{equation}
y_{1(m)} = h_m (s + \eta_1) + \nu_1
\end{equation}
where $h_m$ is the RF fading between (S) and $R_m$, $s$ is the information signal, $\eta_1$ $\backsim$ $\mathcal{CN}$ (0, $\kappa^2_{1}P_1$) is the distortion noise at the source (S), $\kappa_{1}$ is the impairment level parameter in (S), $P_1$ is the average transmitted power from (S) and $\nu_1$  $\backsim$ $\mathcal{CN}$ (0, $\sigma^2_1$) is the AWGN of the RF channel.\\
After reception of the signal $y_{1(m)}$, relay $R_m$ assists with amplification gain $G$ depending on the statistical channel fading of the first hop. According to [1, eq.~(11)], the gain is given by:
\begin{equation}
G^2 \delequal \frac{P_2}{P_1\e{|h_m|^2}(1+\kappa^2_1)+\sigma^2_1}
\end{equation}
where $P_2$ is the average transmitted power from the relay and $\e{.}$ is the expectation operator.\\
The instantaneous SNR of the first hop between (S) and relay $R_m$ can be written as:
\begin{equation}
\gamma_{1(m)} = \frac{|h_m|^2P_1}{\sigma^2_1} = |h_m|^2\mu_1
\end{equation}
where $\mu_1 = \frac{P_1}{\sigma^2_1}$ is the average SNR of the first hop.\\
After the amplification at the relay level, the optical part of the system modulates the received electrical signal by an optical subcarrier using the SIM (Subcarrier Intensity Modulation) technique. The optical signal at the relay $R_m$ is given by:
\begin{equation}
y_{opt(m)} = G(1+\eta_e)y_{1(m)}
\end{equation}
where $\eta_e$ is the electrical-to-optical conversion coefficient.
At the reception (D), our system employs direct detection by eliminating the signal direct component and converts it to the electrical one. The photocurrent after conversion can be expressed as follows:
\begin{equation}
y_{2(m)} = \eta_o I_m ( G(1+\eta_e)(h_m (s + \eta_1) + \nu_1) + \eta_2 ) + \nu_2 
\end{equation}
where $\eta_o$ is the optical-to-electrical conversion, $I_m$ is the optical irradiance between $R_m$ and (D), $\eta_{2}$ $\backsim$ $\mathcal{CN}$ (0, $\kappa^2_{2}P_2$) is the distortion noise at the relay $R_m$, $\kappa_{2}$ is the impairment level parameter at $R_m$ and $\nu_2$  $\backsim$ $\mathcal{CN}$ (0, $\sigma^2_2$) is the AWGN of the optical channel.\\ 
The instantaneous SNR of the second hop between $R_m$ and (D) can be written as:
\begin{equation}
\gamma_{2(m)} = \frac{|I_m|^2\eta_{o}^2P_2}{\sigma^2_2} = |I_m|^2\mu_2
\end{equation}
where $\mu_2 = \frac{\eta_{o}^2P_2}{\sigma^2_2}$ is the average electrical SNR of the second hop. According to eq. (8) in \cite{avg}, the average SNR of the second hop $\e{\gamma_{2(m)}}$ is related to the average electrical SNR $\mu_2$ by the following mathematical expression:
\begin{equation}
\mu_2 = \frac{\alpha\beta\e{\gamma_{2(m)}}}{(\alpha+1)(\beta+1)}
\end{equation}
The end-to-end Signal to Noise plus Distortion Ratio (SNDR) can be expressed as:
\begin{equation}
\gamma_{e2e} = \frac{|h_m|^2 |I_m|^2}{\delta|h_m|^2|I_m|^2 + |I_m|^2(1 + \kappa^2_{2})\frac{\sigma^2_1}{P_1} + \frac{\sigma^2_2}{P_1 G^2}}
\end{equation}
After mathematical manipulations, the end-to-end SNDR is given by:
\begin{equation}
\gamma_{e2e} = \frac{\gamma_{1(m)}\gamma_{2(m)}}{\delta\gamma_{1(m)}\gamma_{2(m)} + (1 + \kappa^2_{2})\gamma_{2(m)} + C}
\end{equation}
where $\delta \delequal \kappa^2_{1} + \kappa^2_{2} + \kappa^2_{1}\kappa^2_{2}$ and the constant $C = \e{\gamma_{1(m)}}(1+\kappa_1^2) + 1$.\\
Since the first hop is subject to Rayleigh fading, the PDF of the relative instantaneous SNR $\gamma_{1(m)}$ is given by [13, eq.~(8)]. After some mathematical manipulations, the PDF of $\gamma_{1(m)}$ is expressed as follows:
\begin{equation}
\begin{split}
f_{\gamma_{1(m)}}(x) = m{N \choose m}\sum_{n=0}^{m-1} {m-1 \choose n} \frac{(-1)^n}{\mu_1}~~~~~\\
\times~\frac{1}{(N-m+n)(1-\rho)+1}~e^{-\frac{(N-m+n+1)x }{[(N-m+n)(1-\rho)+1]\mu_1}}
\end{split}
\end{equation}
The Cumulative Distribution Function (CDF) of $\gamma_{1(m)}$ is given by:
\begin{equation}
F_{\gamma_{1(m)}}(x) = \int_{0}^{x} f_{\gamma_{1(m)}}(t)~dt
\end{equation}
After mathematical manipulations, the CDF of $\gamma_{1(m)}$ can be expressed as follows:
\begin{equation}
\begin{split}
 F_{\gamma_{1(m)}}(x) = 1 - m{N \choose m}\sum_{n=0}^{m-1} {m-1 \choose n} \\ \times \frac{(-1)^n}{N-m+n+1}  e^{-\frac{(N-m+n+1)x }{[(N-m+n)(1-\rho)+1]\mu_1}}
\end{split}
\end{equation}
The constant $C$ that appears in the formula of the end-to-end SNDR depends on $\e{\gamma_{1(m)}}$ which is obtained by:
\begin{equation}
\begin{split}
\e{\gamma_{1(m)}} = \int_{0}^{\infty} x f_{\gamma_{1(m)}}(x)  dx = m{N \choose m} \sum_{n=0}^{m-1} {m-1 \choose n} \\ \times   \frac{(-1)^n[(N-m+n)(1-\rho)+1] \mu_1}{(N-m+n+1)^2}~~~~~~~~
\end{split}
\end{equation}
Regarding the second hop, the optical links are affected by atmospheric turbulence which are modeled by Gamma-Gamma distribution. In this case, the PDF of the instantaneous SNR $\gamma_{2(m)}$ is given by:
\begin{equation}
f_{\gamma_{2(m)}}(x) = \frac{(\alpha\beta)^{\frac{\alpha+\beta}{2}} x^{\frac{\alpha+\beta }{4}-1}}{\Gamma(\alpha)\Gamma(\beta)\mu_2^{\frac{\alpha+\beta}{4}}} K_{\alpha-\beta}\left(2\sqrt{\alpha\beta\sqrt{\frac{x}{\mu_2}}}\right)
\end{equation}
where $K_{\nu}(.)$ is the $\nu$-th order modified Bessel function of the second kind. The parameters $\alpha$ and $\beta$ characterize respectively small-scale and large-scale of scattering process in the atmospheric environment. Hence, the parameters $\alpha$ and $\beta$ can be given by:
\begin{equation}
\alpha = \left( \exp\left[ \frac{0.49 \sigma_R^2}{(1+1.11\sigma_R^{\frac{12}{5}})^{\frac{7}{6}}}\right] -1\right)^{-1}
\end{equation}

\begin{equation}
\beta = \left( \exp\left[ \frac{0.51 \sigma_R^2}{(1+0.69\sigma_R^{\frac{12}{5}})^{\frac{5}{6}}}\right] -1\right)^{-1}
\end{equation}
where $\sigma_R^2$ is called Rytov variance and it is a metric of atmospheric turbulence intensity.
\section{Performance Analysis}
In this section, we provide analysis of the system performance in terms of outage probability and ergodic capacity. This analysis consists of deriving the analytical expressions of outage probability and system capacity and studying their behavior for a high SNR range. Then, we will
show the convergence of the end-to-end SNDR to an accurate ceiling for the case of non-ideal hardware.
\subsection{Outage Probability Analysis}
The outage probability is defined as the probability that the overall instantaneous SNDR falls below a given outage threshold $\gamma_{\text{th}}$. Its formula can be expressed as follows:
\begin{equation}
\begin{split}
P_{\text{out}} \delequal \text{Pr}(\gamma_{e2e} < \gamma_{\text{th}})~~~~~~~~~~~~~~~~~~~~~~~~~~~~~~~~~~~~~~\\ 
= \text{Pr}\left(\frac{\gamma_{1(m)}\gamma_{2(m)}}{\delta\gamma_{1(m)}\gamma_{2(m)} + (1 + \kappa^2_{2})\gamma_{2(m)} + C} < \gamma_{\text{th}}\right)   
\end{split}
\end{equation}
where Pr(.) is the probability notation. After mathematical manipulations, the outage probability can be re-written as:
\begin{equation}
\begin{split}
P_{\text{out}} = \int_{0}^{\infty} \text{Pr}\left(\gamma_{1(m)} < \frac{(1+\kappa^2_{2})\gamma_{\text{th}}\gamma_{2(m)} + C\gamma_{\text{th}}}{(1-\delta\gamma_{\text{th}})\gamma_{2(m)}}~|~\gamma_{2(m)} \right) \\ 
\times~\text{Pr}(\gamma_{2(m)})~d\gamma_{2(m)}~~~~~~~~~~~~~~~~~~~~~~~~~~~
\end{split}
\end{equation}
Since the RF and FSO fadings are independent, the above outage expression can be written as follows:
\begin{equation}
\begin{split}
P_{\text{out}} = \int_{0}^{\infty} F_{\gamma_{1(m)}}\left(\frac{(1+\kappa^2_{2})\gamma_{\text{th}}}{1-\delta\gamma_{\text{th}}} + \frac{C\gamma_{\text{th}}}{(1-\delta\gamma_{\text{th}})\gamma_{2(m)}}\right)~~~~~~~~
\\\times f_{\gamma_{2(m)}}(\gamma_{2(m)})~d\gamma_{2(m)}~~~~~~~~~~~~~~~~~~~~~
\end{split}
\end{equation}
Before deriving the outage probability formula, a \textbf{necessary condition} must be set. This condition is that the outage threshold must be strictly inferior to $\frac{1}{\delta}$, i.e., ($\gamma_{\text{th}}<\frac{1}{\delta}$). Hence, a piecewise form of the outage probability is written as: \\
$$
P_{\text{out}} =
\begin{cases}
\text{eq}.~(20), & \text{if}~\gamma_{\text{th}}<\frac{1}{\delta} \\
1 & \text{otherwise}
\end{cases}
$$

To derive the outage probability formula, we should first substitute eqs. (12) and (14) into eq. (19) and then using the identities in \cite{wolfram}, we finally get the expression in terms of Meijer's G function as shown by eq.~(20):
\begin{equation}
\begin{split}
P_{\text{out}} = 1 - \frac{2^{\alpha+\beta-2}m{N \choose m}}{\pi\Gamma(\alpha)\Gamma(\beta)}\sum_{n=0}^{m-1} {m-1 \choose n}  \frac{(-1)^n}{N-m+n+1}~~~~~~~~~~~~~~~~~~~~~~~~~~~~~~~~~~~~~~\\ \times~\exp\left(-\frac{(N-m+n+1)(1+\kappa^2_{2})\gamma_{\text{th}}}{[(N-m+n)(1-\rho)+1](1-\delta\gamma_{\text{th}})\mu_1}\right)~~~~~~~~~~~~~~~~~~~~~~~~~~~~~~~~~~~~\\\times~\MeijerG[\Bigg]{5}{0}{0}{5}{-}{\frac{\alpha}{2}, \frac{\alpha+1}{2}, \frac{\beta}{2}, \frac{\beta+1}{2}, 0}{\zeta}~~~~~~~~~~~~~~~~~~~~~~~~~~~~~~~~~~~~~~~~~~~~~~~~~~~~~~~~~~~~~~~~
\end{split}
\end{equation}
where $\zeta$ is given by:
\begin{equation}
\zeta = \frac{(\alpha\beta)^2C\gamma_{\text{th}}(N-m+n+1)}{16\mu_1\mu_2(1-\delta\gamma_{\text{th}})[(N-m+n)(1-\rho)+1]}
\end{equation}
Note that the CDF of $\gamma_{1(m)}$ in eq.~(19) is defined only for a positive argument, i.e., ($\gamma_{\text{th}}<\frac{1}{\delta}$), that is why the necessary condition must be set to define the expression of the outage probability. 
\subsection{Ergodic Capacity Analysis}
The channel capacity, expressed by bits/channel use, is obtained by [1, eq.~(32)] as follows:
\begin{equation}
\bar{C} \delequal \frac{1}{2}\e{\log_{2}(1+\gamma_{e2e})}
\end{equation}
Since the transmission is done in two time slots, the factor $\frac{1}{2}$ appears in the channel capacity formula. This formula can be computed by integration referring to the probability density function of $\gamma_{e2e}$. However, deriving a closed form of the channel capacity in our case is very complex if not impossible. To overcome this problem, we should refer to the approximation given by [1, eq.~(35)] 
\begin{equation}
\e{\log_{2}\left(1+\frac{\psi}{\omega}\right)} \approx \log_{2}\left(1 + \frac{\e{\psi}}{\e{\omega}}\right)
\end{equation}
After some mathematical manipulations, the ergodic capacity is expressed as follows:
\begin{equation}
\bar{C} \approx \frac{1}{2} \log_{2}\left(1 + \varrho \right)
\end{equation}
where $\varrho$ is given by:
\begin{equation}
\varrho = \frac{\e{\gamma_{1(m)}}\e{\gamma_{2(m)}}}{\delta\e{\gamma_{1(m)}}\e{\gamma_{2(m)}}+\e{\gamma_{2(m)}}(1+\kappa_2^2)+C}
\end{equation}
The term $\e{\gamma_{2(m)}}$ is defined as:
\begin{equation}
\e{\gamma_{2(m)}} = \frac{(\alpha+1)(\beta+1)\mu_2}{\alpha\beta}    
\end{equation}
To quantify the ergodic capacity, it is possible to derive a closed-form of an upper bound using the following theorem.
\newtheorem{theorem}{Theorem}
\begin{theorem}
For Asymmetric (Rayleigh/Gamma-Gamma) fading channels, the ergodic capacity $\bar{C}$ with AF relaying and non-ideal hardware has an upper bound defined by:
\begin{equation}
    \bar{C} \leq \frac{1}{2}\log_{2}\left(1 + \frac{\mathcal{J}}{\delta\mathcal{J} + 1}\right)
\end{equation}
\end{theorem}
where $\mathcal{J}$ is given by:
\begin{equation}
 \mathcal{J} = \e{\Theta},~~\Theta = \frac{\gamma_{1(m)}\gamma_{2(m)}}{(1+\kappa_2^2)\gamma_{2(m)} + C}   
\end{equation}
Using the following identities given by \cite{27}, eq.~(24) and \cite{26}, eq.~(2.24.1.1) and after some mathematical manipulations, $\mathcal{J}$ can be derived as follows:
\begin{equation}
\begin{split}
\mathcal{J} = \frac{\e{\gamma_{1(m)}}}
{4\pi\Gamma(\alpha)\Gamma(\beta)(1+\kappa_2^2)}\left[\frac{C(\alpha\beta)^2}{(1+\kappa_2^2)\mu_2}\right]^{\frac{\alpha+\beta}{4}}~~~~~~~~~~~~~~~~~~\\ \times~
\MeijerG[\Bigg]{5}{1}{1}{5}{-\frac{\alpha+\beta}{4}}{\frac{\alpha-\beta}{4}, \frac{\alpha-\beta+2}{4}, \frac{\beta-\alpha}{4}, \frac{\beta-\alpha+2}{4}, -\frac{\alpha+\beta}{4}}{\frac{C(\alpha\beta)^2}{16\mu_2(1+\kappa_2^2)}}
\end{split}
\end{equation}
\subsection{Asymptotic Analysis}
In this subsection, we study the behavior of the SNDR and the ergodic capacity for a high SNR regime.\\
For the ideal case, the SNDR can be written as:
\begin{equation}
\gamma_{e2e} = \frac{\gamma_{1(m)}\gamma_{2(m)}}{\gamma_{2(m)} + \e{\gamma_{1(m)}} + 1}
\end{equation}
For a high SNR range, the SNDR is not upper bounded as shown below:
\begin{equation}
\lim_{\mu_1,\mu_2\to\infty} \gamma_{e2e} = \infty
\end{equation}
We can see that for the ideal case, i.e., without hardware impairments, SNDR has no ceiling which means that the performance is not limited to a high-rate system.\\ 
Considering the case of hardware impairments, the SNDR is upper bounded by a ceiling at the high SNR regime which is shown as follows:
\begin{equation}
\lim_{\mu_1,\mu_2\to\infty}\gamma_{e2e} = \frac{1}{\delta}
\end{equation}
We observe that the SNDR converges to a constant ceiling $\gamma^* = \frac{1}{\delta}$.
\newtheorem{thm}{Theorem}
\newtheorem{cor}[thm]{Corollary}
\begin{cor}
Assuming that $\mu_1$, $\mu_2$ largely grow up and the electrical and optical channels are mutually independant, the system capacity converges to a ceiling given by $C^* = \frac{1}{2}\log_2(1+\gamma^*)$. 
\end{cor}
\begin{proof}
Since the SNDR converges to $\gamma^*$, the dominated convergence theorem permits to move the limit inside the logarithm function as follows:
\begin{equation}
\begin{split}
\lim_{\mu_1,\mu_2\to\infty}\frac{1}{2}\log_2(1 + \e{\gamma_{e2e}}) = \frac{1}{2} \log_2(1 + \lim_{\mu_1,\mu_2\to\infty}\e{\gamma_{e2e}})\\
= \frac{1}{2}\log_2(1+\gamma^*)~~~~~~~~~~~~~~~~~
\end{split}
\end{equation}
\end{proof}
Note that for the ideal case, the ergodic capacity grows infinitely without an upper bound for a high SNR regime and so we can see clearly the impairment limitations effect on the system performance.
\section{Numerical results}
This section presents numerical and simulation results of the outage probability and the ergodic capacity dependence on the system parameters. We use Monte-Carlo simulation to validate the numerical results.
\vspace*{-0.5cm}
\begin{center}
\includegraphics[width=9.7cm,height=7cm]{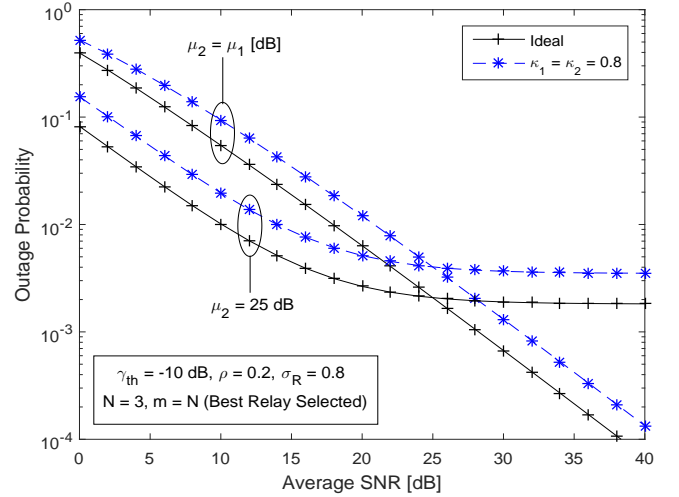}
\captionof{figure}{Outage probability versus the average SNR for different values of impairment ($\kappa_1, \kappa_2$) and average electrical SNR of the optical channel.}
\label{fig1}
\end{center}
\vspace*{-0.5cm}
Fig. 1 shows the outage probability dependence on the average SNR for different impairment values and average electrical SNR of the optical channel. If we increase the impairment values, the outage performance deteriorates and if we decrease the impairment values, the outage performance improves. Assuming $\mu_2$ = 25 dB, we note that increasing the average SNR of the RF link results in the existence of the outage floor, i.e., increasing the RF SNR  has no effect on the outage probability. However, if we assume $\mu_2$ grows simultaneously with $\mu_1$, the system performance improves better after the 25 dB value and the outage floor disappears.
\begin{center}
\includegraphics[width=9.7cm,height=7cm]{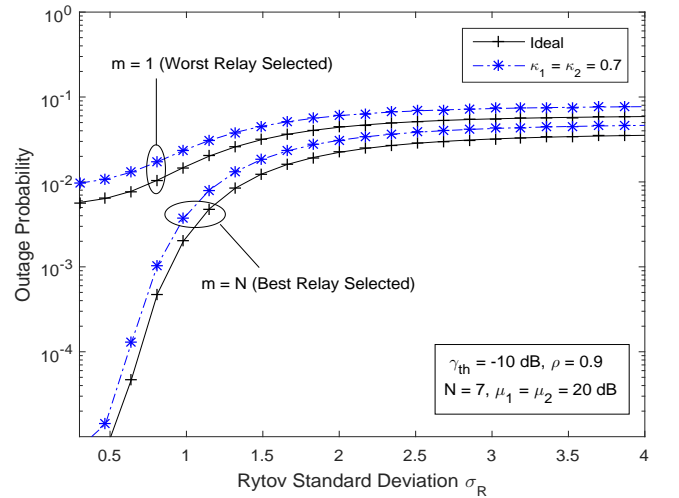}
\captionof{figure}{Outage probability versus $\sigma_R$ for different values of impairment ($\kappa_1$, $\kappa_2$) and for the best and worst scenarios.}
\label{fig1}
\end{center}

The variations of the outage probability with respect to Rytov standard deviation are shown in Fig. 2. Both, the best and worst relay selection scenarios are presented for different impairment values. For weak turbulence conditions, the outage probability is better when the best relay is available to forward the communication. For this case, decreasing the impairment levels results in further improvement of the outage performance. If we select the worst relay and decrease the impairment levels, the outage probability does not significantly improve in comparison with the outage performance of the best relay selection scenario. However, either selecting the best relay or decreasing the impairment levels has no effect on the system performance for strong turbulence conditions.
\vspace*{-0.5cm}
\begin{center}
\includegraphics[width=9.7cm,height=7cm]{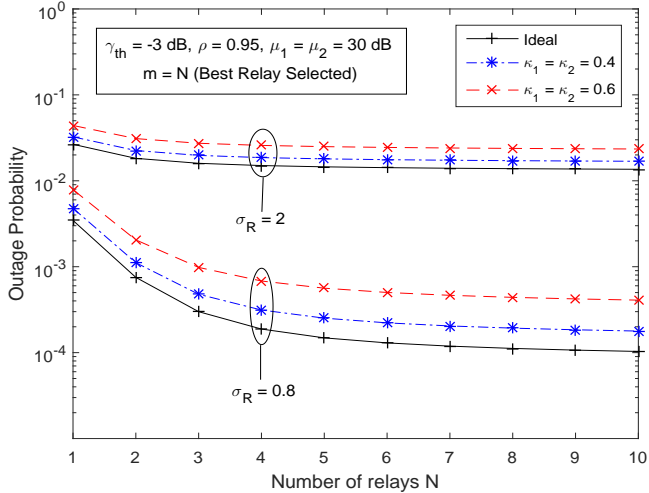}
\captionof{figure}{Outage probability versus the number of relays $N$ for different values of impairment ($\kappa_1$, $\kappa_2$) in different atmospheric turbulence conditions.}
\label{fig1}
\end{center}
\vspace*{-1cm}
\begin{center}
\includegraphics[width=9.7cm,height=7cm]{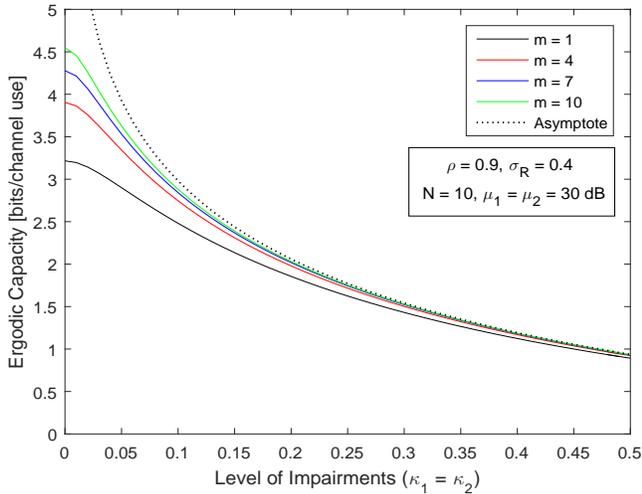}
\captionof{figure}{Ergodic capacity versus the level of impairments ($\kappa_1$, $\kappa_2$) for different ranks of the selected relay.}
\label{fig1}
\end{center}
Fig. 3 shows the outage probability dependence on the number of relays for different values of impairment. The performance is presented for weak and moderate atmospheric tubulence conditions. We observe that for weak turbulence conditions, the outage probability improves when increasing the number of relays. In addition, the number of relays can have a significant role at low impairment levels. However, for moderate turbulence conditions, either increasing the number of relays or decreasing the impairment levels has no significant effect on the system performance compared to the performance under weak turbulence conditions.
\vspace*{-0.5cm}
\begin{center}
\includegraphics[width=9.7cm,height=7cm]{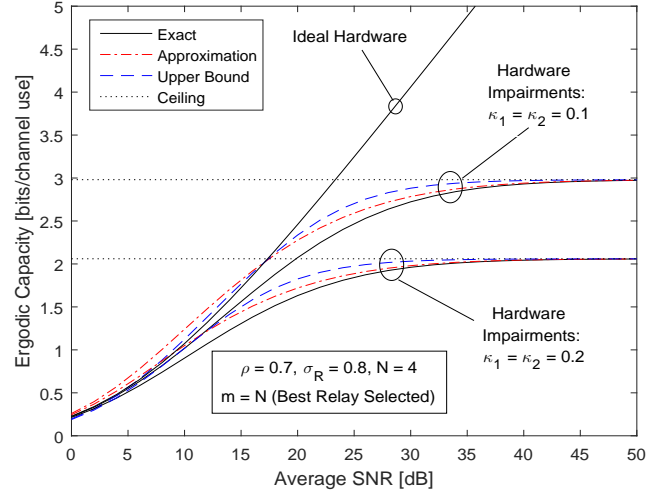}
\captionof{figure} {Ergodic capacity versus the average SNR for different values of impairment ($\kappa_1$, $\kappa_2$).}
\label{fig1}
\end{center}
\vspace*{-0.5cm}
Fig. 4 shows the ergodic capacity dependence on the levels of impairment for the best and worst relay selection scenarios. For low impairment levels, the best relay selection scenario shows a higher performance compared to the worst selection scenario. However, increasing the impairment levels results in neglecting the effect of the best selection scenario and thus leading to the same performance for both scenarios.
\begin{center}
\includegraphics[width=9.7cm,height=7cm]{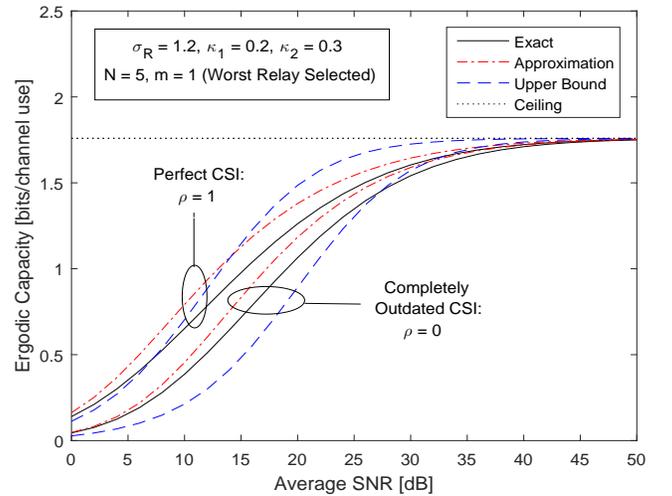}
\captionof{figure}{Ergodic capacity versus the average SNR for different values of the time correlation $\rho$.}
\label{fig1}
\end{center}
Fig. 5 shows the dependence of the ergodic capacity on the average SNR for different impairment values. As expected, the system performs better under low impairment levels. For the ideal case, we observe that the capacity grows infinitely without an upper bound contrary to the non-ideal case where the capacity is limited by a capacity ceiling $C^{*}$ eq.~(29) for a high SNR regime.\\
 If we assume a perfect CSI estimation ($\rho$ = 1), the complete correlation between the actual CSI and the CSI employed for the relay selection is achieved, i.e., the selection of the best relay is certainly achieved, and hence the system performs better and achieves high performance. Whereas, when the CSI used for the relay selection is completely outdated ($\rho$ = 0), this CSI and the CSI used for transmitting the signal are fully uncorrelated and hence there is a low probability that the selected relay is the best one leading to deteriorating the system performance.
\section{Conclusion}
In this work, we introduced an aggregate model of impairments to a mixed RF/FSO system of multiple relays. We also used the Amplify-and-Forward scheme with a fixed gain relaying. In addition, we proposed the PRS protocol assuming that the best relay is not always available to forward the signal. Moreover, we investigate the effects of the outdated CSI, the average electrical SNR of the optical link, the rank of the selected relay, the number of relays and the impairment levels under weak, moderate and strong atmospheric turbulence conditions. The results showed that the system performance depends significantly on the state of the optical link in terms of the optical SNR and atmospheric tubulence conditions. We also proved the existence of the ceilings for the effective SNDR and the capacity which limit the system performance, i.e., the system saturates by the hardware impairments contrary to the ideal case where the system performance grows infinitely without limits. In the future, we plan to extend this work following three major research directions. First, we will consider a more complex optical fading that takes into account the pointing errors and the atmospheric path loss. Second, it is important to study the different relaying protocols and their effects on the system performance. Finally, deriving analytical expressions of the Symbol Error Rate (SER) for different types of modulation is important to get an accurate analysis of the system.
\bibliographystyle{IEEEtran}

\bibliography{main}

\end{document}